\documentclass{article} 

\newcommand{\xpt}{\edef\f@size{\@xpt}\rm}
\usepackage{nips14submit_e,times}
\usepackage{url}
\usepackage{fullpage}
\newcommand{\ttbf}{\ttfamily\bfseries}
\newcommand{\bm}{\mathbf}
\usepackage{longtable}
\usepackage{rotating}
\usepackage{amsmath,amssymb,amsthm}
\usepackage{natbib}
\usepackage{caption}
\usepackage{subcaption}

\DeclareMathOperator*{\minimize}{\text{minimize}}
\DeclareMathOperator*{\maximize}{\text{maximize}}

\DeclareMathOperator*{\argmin}{arg\,min}
\DeclareMathOperator{\st}{\text{subject to}\ }

\DeclareMathOperator{\tr}{\text{tr}}

\makeatletter
\renewcommand*\env@matrix[1][c]{\hskip -\arraycolsep
  \let\@ifnextchar\new@ifnextchar
  \array{*\c@MaxMatrixCols #1}}
\makeatother

\usepackage{algorithm}
\usepackage{algorithmic}
\usepackage{bbm}
\usepackage{footnote} 

\usepackage{caption}
\usepackage[labelformat=simple]{subcaption}

\usepackage{amsmath,amsfonts,amssymb}
\usepackage{float}
\usepackage{multirow}

\makeatletter
\@ifundefined{knitrout}{%
  \newenvironment{knitrout}{}{} 
  \definecolor{fgcolor}{rgb}{0.345, 0.345, 0.345}
  \def\maxwidth{ %
    \ifdim\Gin@nat@width>\linewidth
    \linewidth
    \else
    \Gin@nat@width
    \fi
  }
}{}
\makeatother

\newcommand{\con}{\mbox{\scriptsize con}}

\newcommand{\la}{\langle}
\newcommand{\ra}{\rangle}

\newcommand{\D}{_\text{\scriptsize D}}
\newcommand{\X}{_\text{\scriptsize X}}

\newtheorem{theorem}{Theorem}[section]

\newcommand{\RR}{I\!\!R} 

\title{Optimization Methods for Sparse Pseudo-Likelihood Graphical Model Selection}

\author{
Sang-Yun Oh \\
Computational Research Division\\
Lawrence Berkeley National Lab \\
\texttt{syoh@lbl.gov}
\And
Onkar Dalal \\
Stanford University \\
\texttt{onkar@alumni.stanford.edu} \\
\AND
Kshitij Khare \\
Department of Statistics \\
University of Florida \\
\texttt{kdkhare@stat.ufl.edu} \\
\And
Bala Rajaratnam \\
Department of Statistics \\
Stanford University \\
\texttt{brajarat@stanford.edu} \\
}

\nipsfinalcopy 

\begin{document}

\maketitle

\begin{abstract}
  Sparse high dimensional graphical model selection is a popular topic
  in contemporary machine learning. To this end, various useful
  approaches have been proposed in the context of $\ell_1$-penalized
  estimation in the Gaussian framework. Though many of these inverse
  covariance estimation approaches are demonstrably scalable and have
  leveraged recent advances in convex optimization, they still depend
  on the Gaussian functional form. To address this gap, a convex
  pseudo-likelihood based partial correlation graph estimation method
  (CONCORD) has been recently proposed. This method uses 
  coordinate-wise minimization of a regression based
  pseudo-likelihood, and has been shown to have robust model selection
  properties in comparison with the Gaussian approach. In direct
  contrast to the parallel work in the Gaussian setting however, this
  new convex pseudo-likelihood framework has not leveraged the
  extensive array of methods that have been proposed in the machine
  learning literature for convex optimization. In this paper, we
  address this crucial gap by proposing two proximal gradient methods
  (CONCORD-ISTA and CONCORD-FISTA) for performing $\ell_1$-regularized
  inverse covariance matrix estimation in the pseudo-likelihood
  framework. We present timing comparisons with coordinate-wise
  minimization and demonstrate that our approach yields tremendous
  payoffs for $\ell_1$-penalized partial correlation graph estimation
  outside the Gaussian setting, thus yielding the fastest and most
  scalable approach for such problems. We undertake a theoretical
  analysis of our approach and rigorously demonstrate convergence, and
  also derive rates thereof.
\end{abstract}

\section{Introduction}

\subsection{Background} 
Sparse inverse covariance estimation has received tremendous attention
in the machine learning, statistics and optimization
communities.  These sparse models, popularly known as
graphical models, have widespread use in various applications,
especially in high dimensional settings. The most popular inverse
covariance estimation framework is arguably the $\ell_1$-penalized
Gaussian likelihood optimization framework as given by
\begin{align*}
  \minimize_{\Omega\in\mathbf{S}^p_{++}} &\quad -\log\det \Omega + \tr(S\Omega) + \lambda\|\Omega\|_1
\end{align*}
where $\mathbf{S}^p_{++}$ denotes the space of $p$-dimensional
positive definite matrices, and $\ell_1$-penalty is imposed on the
elements of $\Omega=(\omega_{ij})_{1\leq i\leq j\leq p}$ by the term
$\|\Omega\|_1=\sum_{i,j}|\omega_{ij}|$ along with the scaling factor
$\lambda>0$. The matrix $S$ denotes the sample covariance matrix of
the data ${\bf Y}\in\RR^{n\times p}$. As the $\ell_1$-penalized
log likelihood is convex, the problem becomes more tractable and has
benefited from advances in convex optimization. Recent efforts in the
literature on Gaussian graphical models therefore have focused on
developing principled methods which are increasingly more and more
scalable. The literature on this topic is simply enormous and for the
sake of brevity, space constraints and the topic of this paper, we
avoid an extensive literature review by referring to the references in
the seminal work of \cite{Banerjee2008} and the very recent work of
\cite{Dalal2014}. These two papers contain references to recent work,
including past NIPS conference proceedings.

\subsection{The CONCORD method}

Despite their tremendous contributions, one shortcoming of the
traditional approaches to $\ell_1$-penalized likelihood maximization
is the restriction to the Gaussian  assumption. To
address this gap, a number of $\ell_1$-penalized pseudo-likelihood
approaches have been proposed: SPACE \cite{Peng2009} and SPLICE
\cite{Rocha2008}, SYMLASSO \cite{Friedman2010}. These approaches are
either not convex, and/or convergence of corresponding maximization
algorithms are not established. In this sense, non-Gaussian partial
correlation graph estimation methods have lagged severely behind,
despite the tremendous need to move beyond the Gaussian framework for
obvious practical reasons. In very recent work, a convex
pseudo-likelihood approach with good model selection properties called
CONCORD \cite{Khare2013} was proposed. The CONCORD algorithm minimizes
\begin{align}
  Q_{\con} (\Omega) &= -\sum_{i=1}^p n\log \omega_{ii} + \frac{1}{2}
  \sum_{i=1}^p \| \omega_{ii} {\bf Y}_i + \sum_{j \neq i} \omega_{ij}
  {\bf Y}_j \|_2^2 + n\lambda \sum_{1 \leq i < j \leq p} |\omega_{ij}|
  \label{eq:concord regression form}
\end{align}
via cyclic coordinate-wise descent that alternates between updating
off-diagonal elements and diagonal elements. It is straightforward to
show that operators $T_{ij}$ for updating $(\omega_{ij})_{1\leq
  i<j\leq p}$ (holding $(\omega_{ii})_{1\leq i \leq p}$ constant) and
$T_{ii}$ for updating $(\omega_{ii})_{1\leq i \leq p}$ (holding
$(\omega_{ij})_{1\leq i<j\leq p}$ constant) are given by
\begin{align}
  \left(T_{ij} (\Omega)\right)_{ij} &=
  \frac{S_{\lambda}\left(-\left(\sum_{j' \neq j} \omega_{ij'} s_{jj'}
        + \sum_{i' \neq i} \omega_{i'j}
        s_{ii'}\right)\right)} {s_{ii} + s_{jj}}\\
  \left( T_{ii} (\Omega) \right)_{ii} &= \frac{- \sum_{j \neq i}
    \omega_{ij} s_{ij} + \sqrt{\left( \sum_{j \neq i} \omega_{ij}
        s_{ij} \right)^2 + 4 s_{ii}}}{2 s_{ii}}.
\end{align}
This coordinate-wise algorithm is shown to converge to a global minima
though no rate is given \citep{Khare2013}. Note that the equivalent
problem assuming a Gaussian likelihood has seen much development in
the last ten years, but a parallel development for the recently
introduced CONCORD framework is lacking for obvious reasons. We
address this important gap by proposing state-of-the-art proximal
gradient techniques to minimize $Q_{\con}$. A rigorous theoretical
analysis of the pseudo-likelihood framework and the associated
proximal gradient methods which are proposed is undertaken. We
establish rates of convergence and also demonstrate that our approach
can lead to massive computational speed-ups, thus yielding extremely
fast and principled solvers for the sparse inverse covariance
estimation problem outside the Gaussian setting.


\section{CONCORD using proximal gradient methods}

The penalized matrix version the CONCORD objective function in
\eqref{eq:concord regression form} is given by
\begin{align}
  Q_{\con}(\Omega) = \frac{n}{2}\left[ -\log | \Omega_D^2 | +
    \tr(\mathbf{S}\Omega^2) + \lambda \|\Omega\X\|_1
  \right].\label{eq:matrix concord}
\end{align}
where $\Omega\D$ and $\Omega\X$ denote the diagonal and off-diagonal
elements of $\Omega$. We will use the notation $A = A\D + A\X$ to
split any matrix $A$ into its diagonal and off-diagonal terms.

This section proposes a scalable and thorough approach to solving the
CONCORD objective function using recent advances in convex
optimization and derives rates of convergence for such algorithms. In
particular, we use proximal gradient-based methods to achieve this
goal and demonstrate the efficacy of such methods for the non-Gaussian
graphical modeling problem. First, we propose CONCORD-ISTA and
CONCORD-FISTA in section \ref{proximalgradient}: methods which are
inspired by the iterative soft-thresholding algorithms in
\cite{beck2009fast}. We undertake a comprehensive treatment of the
CONCORD optimization problem by also investigating the dual of the
CONCORD problem. Other popular methods in the literature, including
the potential use of alternating minimization algorithm and the second
order proximal Newton’s method CONCORD-PNOPT, are considered in
Supplemental section \ref{othermethods}.


\subsection{Iterative Soft Thresholding Algorithms: CONCORD-ISTA, CONCORD-FISTA}
\label{proximalgradient}
The iterative soft-thresholding algorithms (ISTA) have recently gained
popularity after the seminal paper by Beck and Teboulle
\cite{beck2009fast}. The ISTA methods are based on the
Forward-Backward Splitting method from \cite{rockafellar1976monotone}
and Nesterov's accelerated gradient methods \cite{nesterov1983method}
using soft-thresholding as the proximal operator for the
$\ell_1$-norm. The essence of the proximal gradient algorithms is to
divide the objective function into a smooth part and a non-smooth
part, then take a proximal step (w.r.t. the non-smooth part) in the
negative gradient direction of the smooth part. Nesterov's accelerated
gradient extension \cite{nesterov1983method} uses a combination of
gradient and momentum steps to achieve accelerated rates of
convergence. In this section, we apply these methods in the context of
CONCORD which also has a composite objective function.

The matrix CONCORD objective function \eqref{eq:matrix concord} can be split into a smooth part $h_1(\Omega)$ and a non-smooth part $h_2(\Omega)$:
\begin{align}
  h_1(\Omega) &= -\log\det{\Omega\D} +
  \frac{1}{2}\tr(\Omega S \Omega),\,\,\,\,\,
  h_2(\Omega) = \lambda \|\Omega\X\|_{1}.
\end{align}
The gradient and hessian of the smooth function $h_1$ are given by
\begin{align}
 \nabla{h_1}(\Omega) &= {\Omega_D}^{-1} + \frac{1}{2} \left( S \Omega^{T} + \Omega S \right),\nonumber \\ 
\nabla^{2}{h_1}(\Omega) &= \sum_{i=1}^{i=p} {\omega_{ii}^{-2} \left[{e_i}{e_i}^{T} \otimes {e_i}{e_i}^{T}\right]} + \frac{1}{2} \left( S \otimes I + I \otimes S \right),\label{hessian} 
\end{align}
where $e_i$ is a column vector of zeros except for a one in the $i$-th
position.
\begin{figure}
  \centering
  \begin{minipage}[t]{0.495\textwidth}
\begin{algorithm}[H]
  \caption{CONCORD-ISTA}
  \label{concordmatrix:fbs}

  \begin{algorithmic}

    \STATE Input: sample covariance matrix $S$, penalty matrix
    $\Lambda$

    \STATE Initialize: ${\Omega}^{(0)} \in \mathbb{S}^{p}_{+}$,
    $\tau_{(0,0)} \leq 1$,\\
    \qquad $c<1$, $\Delta_{\texttt{subg}} =
    2\epsilon_{\texttt{subg}}$.


    \WHILE{$\Delta_{\texttt{subg}} > \epsilon_{\texttt{subg}}$
      \textbf{or} $\Delta_{\texttt{func}} > \epsilon_{\texttt{func}}$}
    \STATE \textit{Compute $\nabla{h_1}$:} 

    \STATE \qquad $G^{(k)} = -\left(\Omega\D^{(k)}\right)^{-1} +
    \frac{1}{2}\left(S\,\,\Omega^{(k)} + \Omega^{(k)}S\right)$

    \STATE \textit{Compute} $\tau_{k}$: 

    \STATE \qquad Largest $\tau_k \in
    \{c^{j}\tau_{(k,0)}\}_{j=0,1,\ldots}$ such that,
    \STATE \qquad $\Omega^{(k+1)} = {\cal
      S}_{\tau_k\Lambda}\displaystyle\left(\Omega^{(k)} -
      {\tau_k}{G}^{(k)} \right)$ satisfies \eqref{suffdescmat}.

    \STATE \textit{Update:} $\Omega^{(k+1)}$ using the appropriate step size.

    \STATE \textit{Compute next initial step size: $\tau_{(k+1,0)}$}


    \STATE \textit{Compute convergence criteria:} 

    \STATE \qquad $\Delta_{\texttt{subg}} =
    \displaystyle\frac{\| \nabla{h_1}(\Omega^{(k)})
      + \partial{h_2}(\Omega^{(k)}) \|}{\|\Omega^{(k)}\|}$.

    \STATE

    \ENDWHILE
  \end{algorithmic}
\end{algorithm}
  \end{minipage}
  \begin{minipage}[t]{0.495\textwidth}
\begin{algorithm}[H]
  \caption{CONCORD-FISTA}
  \label{concordmatrix:fista}
  \begin{algorithmic}

    \STATE Input: sample covariance matrix $S$, penalty matrix
    $\Lambda$

    \STATE Initialize: $({\Theta}^{(1)} = ){\Omega}^{(0)} \in
    \mathbb{S}^{p}_{+}$, $\alpha_{1} = 1$, $\tau_{(0,0)} \leq 1$,\\
    \qquad $c<1$, $\Delta_{\texttt{subg}} = 2\epsilon_{\texttt{subg}}$.


    \WHILE{$\Delta_{\texttt{subg}} > \epsilon_{\texttt{subg}}$
      \textbf{or} $\Delta_{\texttt{func}} > \epsilon_{\texttt{func}}$}

    \STATE \textit{Compute $\nabla{h_1}$:}

    \STATE \qquad $G^{(k)} = -\left(\Theta^{(k)}\D\right)^{-1}
    + \frac{1}{2}\left(S {\Theta^{(k)}} + {\Theta^{(k)}} S\right)$

    \STATE \textit{Compute} $\tau_{k}$: 

    \STATE \qquad Largest $\tau_k \in
    \{c^{j}\tau_{(k,0)}\}_{j=0,1,\ldots}$ such that,

    \STATE \qquad $\Omega^{(k)} = {\cal
      S}_{\tau_k\Lambda}\displaystyle\left(\Theta^{(k)} -
      {\tau_k}{G}^{(k)} \right)$ satisfies \eqref{suffdescmat}


    \STATE \textit{Update:} $\alpha_{k+1} = \displaystyle( 1 + \sqrt{1
      + 4{\alpha_{k}}^{2}} ) / 2$

    \STATE \textit{Update:} $\Theta^{(k+1)} = \Omega^{(k)} +
    \left(\frac{\alpha_{k} - 1}{\alpha_{k+1}}\right)\left(\Omega^{(k)} -
      \Omega^{(k-1)}\right)$

    \STATE \textit{Compute next initial step size: $\tau_{(k+1,0)}$}



    \STATE \textit{Compute convergence criteria:}

    \STATE \qquad $\Delta_{\texttt{subg}} =
    \displaystyle\frac{\| \nabla{h_1}(\Omega^{(k)})
      + \partial{h_2}(\Omega^{(k)}) \|}{\|\Omega^{(k)}\|}$.


    \ENDWHILE
  \end{algorithmic}
\end{algorithm}
  \end{minipage}
\end{figure}
The proximal operator for
the non-smooth function $h_{2}$ is given by element-wise
soft-thresholding operator ${\cal S}_{\lambda}$ as
\begin{align}
  \mbox{prox}_{h_2}(\Omega) &= \argmin_{\Theta} \left\{h_2(\Theta) + \frac{1}{2}\|\Omega - \Theta\|_{F}^{2}\right\} \nonumber \\
  &= {\cal S}_{\Lambda}(\Omega) = \texttt{sign}(\Omega)\max\{|\Omega|-\Lambda,0\}, \label{proxopelement} 
\end{align}
where $\Lambda$ is a matrix with $0$ diagonal and $\lambda$ for each
off-diagonal entry. The details of the proximal gradient algorithm
CONCORD-ISTA are given in Algorithm \ref{concordmatrix:fbs}, and the
details of the accelerated proximal gradient algorithm CONCORD-FISTA
are given in Algorithm \ref{concordmatrix:fista}.

\subsection{Choice of step size}
In the absence of a good estimate of the Lipschitz constant $L$, the
step size for each iteration of CONCORD-ISTA and CONCORD-FISTA is
chosen using backtracking line search. The line search for iteration
$k$ starts with an initial step size $\tau_{(k,0)}$ and reduces the
step with a constant factor $c$ until the new iterate satisfies the
sufficient descent condition:
\begin{align}
{h_1}({\Omega}^{(k+1)}) \leq {\cal Q}(\Omega^{(k+1)}, \Omega^{(k)})\label{suffdescmat}
\end{align}
where, 
\begin{align*}
{\cal Q}(\Omega, \Theta) = {h_1}({\Theta}) + \tr\left(({\Omega} - {\Theta})^{T}\nabla{h_1}({\Theta})\right) + \frac{1}{2\tau} \big\|{\Omega} - {\Theta}\big\|^{2}_{F}.
\end{align*}
In section \ref{sect: numerical}, we have implemented algorithms choosing the initial step size in three different ways: (a) a constant starting step size (=1), (b) the feasible step size from the previous iteration $\tau_{k-1}$, (c) the step size
heuristic of Barzilai-Borwein. The Barzilai-Borwein heuristic step size is given by
\begin{align}
\tau_{k+1,0} = \displaystyle\frac{\tr\left( (\Omega^{(k+1)}-\Omega^{(k)})^{T}(\Omega^{(k+1)}-\Omega^{(k)}) \right)}{ \tr\left( (\Omega^{(k+1)}-\Omega^{(k)})^{T}(G^{(k+1)}-G^{(k)}) \right)}.
\label{bb-step}
\end{align}
This is an approximation of the secant equation which works as a proxy for second order information using successive gradients (see \cite{barzilai1988two} for details).

\subsection{Computational complexity}
After the one time calculation of $S$, the most significant
computation for each iteration in CONCORD-ISTA and CONCORD-FISTA
algorithms is the matrix-matrix multiplication $W = S\Omega$ in the
gradient term. If $s$ is the number of non-zeros in $\Omega$, then $W$
can be computed using ${\cal O}(sp^2)$ operations if we exploit the
extreme sparsity in $\Omega$. The second matrix-matrix multiplication
for the term $\tr(\Omega(S\Omega))$ can be computed efficiently using
$\tr(\Omega W) = \sum \omega_{ij} w_{ij}$ over the set of non-zero
$\omega_{ij}$'s. This computation only requires ${\cal O}(s)$
operations. The remaining computations are all at the element level
which can be completed in ${\cal O}(p^2)$ operations. Therefore, the
overall computational complexity for each iteration reduces to ${\cal
  O}(sp^2)$. On the other hand, the proximal gradient algorithms for the Gaussian
framework require inversion of a full $p\times p$ matrix which is
non-parallelizable and requires ${\cal O}(p^3)$ operations. The
coordinate-wise method for optimizing CONCORD in \cite{Khare2013} also
requires cycling through the $p^2$ entries of $\Omega$ in specified
order and thus does not allow parallelization. In contrast,
CONCORD-ISTA and CONCORD-FISTA can use `perfectly parallel'
implementations to distribute the above matrix-matrix
multiplications. At no step do we need to keep all of the dense
matrices $S, S\Omega, \nabla{h_{1}}$ on a single machine. Therefore,
CONCORD-ISTA and CONCORD-FISTA are scalable to any high dimensions
restricted only by the number of machines.

\section{Convergence Analysis}
\label{sect: convergence}
In this section, we prove convergence of CONCORD-ISTA and
CONCORD-FISTA methods along with their respective convergence rates of
${\cal O}(1/k)$ and ${\cal O}(1/k^2)$. We would like to point out
that, although the authors in \cite{Khare2013} provide a proof of
convergence for their coordinate-wise minimization algorithm for
CONCORD, they do not provide any rates of convergence. The arguments
for convergence leverage the results in \cite{beck2009fast} but
require some essential ingredients. We begin with proving lower and
upper bounds on the diagonal entries $\omega_{kk}$ for $\Omega$
belonging to a level set of $Q_{\con} (\Omega)$. The lower bound on
the diagonal entries of $\Omega$ establishes Lipschitz continuity of
the gradient $\nabla{h_1}(\Omega)$ based on the hessian of the smooth
function as stated in \eqref{hessian}. The proof for the lower bound
uses the existence of an upper bound on the diagonal entries. Hence,
we prove both bounds on the diagonal entries. We begin by defining a
level set ${\cal C}_0$ of the objective function starting with an
arbitrary initial point $\Omega^{(0)}$ with a finite function value as
\begin{alignat}{1}
  \label{levelset} {\cal C}_{0} = \left\{\Omega \,\,\, | \,\,\,
    Q_{\con}(\Omega) \leq Q_{\con}(\Omega^{(0)}) = M \right\}.
\end{alignat}
For the positive semidefinite matrix $S$, let $U$ denote
$\frac{1}{\sqrt{2}}$ times the upper triangular matrix from the LU
decomposition of $S$, such that $S = 2U^{T}U$ (the factor $2$
simplifies further arithmetic). Assuming the diagonal entries of $S$
to be strictly nonzero (if $s_{kk} = 0$, then the $k^{th}$ component
can be ignored upfront since it has zero variance and is equal to a
constant for every data point), we have at least one $k$ such that
$u_{ki} \neq 0$ for every $i$. Using this, we prove the following
theorem bounding the diagonal entries of $\Omega$.
\begin{theorem}
\label{up-low-bound}
For any symmetric matrix $\Omega$ satisfying $\Omega \in {\cal C}_0$,
the diagonal elements of $\Omega$ are bounded above and below by
constants which depend only on $M$, $\lambda$ and $S$. In other words,
\[
0 < a_{M,\lambda,S} \leq |\omega_{kk}| \leq b_{M,\lambda,S}, \,\,\,
\forall \,\,\, k = 1, 2, \ldots, p,
\]
for some constants $a_{M,\lambda,S}$ and $b_{M,\lambda,S}$.
\end{theorem}
\begin{proof}
(a) Upper bound: Suppose $|\omega_{ii}| = \max\{|\omega_{kk}|, \mbox{for}\,\, k = 1, 2, \ldots, p\}$. Then, we have
\begin{align}
M &= Q_{\con}(\Omega^{(0)}) \geq Q_{\con}(\Omega) = h_1(\Omega) + h_2(\Omega) \nonumber \\
& \geq -\log\det{\Omega_D} + \tr\left((U\Omega)^{T}(U\Omega)\right) + \lambda\|\Omega_X\|_{1} \nonumber \\
& = -\log\det{\Omega_D} + \|U\Omega\|_{F}^{2} + \lambda\|\Omega_X\|_{1}.
\end{align}
Considering only the ${ki}^{th}$ entry in the Frobenious norm term and
the $i^{th}$ column penalty in the third term we get
\begin{align}
M & \geq -p\log|\omega_{ii}| + \left(\sum_{j=k}^{j=p}u_{kj}\omega_{ji}\right)^{2} + \lambda\sum_{j=k, j\neq i}^{j=p}|\omega_{ji}|.
\label{intermediate}
\end{align}
Now, suppose $|u_{ki}\omega_{ii}| = z$ and $\sum_{j=k, j\neq i}^{j=p}u_{kj}\omega_{ji} = x$. Then 
\begin{align}
|x| &\leq \displaystyle\sum_{j=k, j\neq i}^{j=p}|u_{kj}||\omega_{ji}| \leq \bar{u}\displaystyle\sum_{j=k, j\neq i}^{j=p}|\omega_{ji}|, \nonumber
\end{align}
where $\bar{u} = \max\{|u_{kj}|, \mbox{for}\,\, j = k, k+1, \ldots, p \,\, \mbox{and}\,\, j\neq i\}$. Going back to the inequality (\ref{intermediate}), for $\bar{\lambda} = \frac{\lambda}{2\bar{u}}$, we have
\begin{align}
\bar{M} = M + \bar{\lambda}^2 - p \log|u_{ki}| & \geq -p\log z + \left(z + x\right)^{2} + 2\bar{\lambda}|x| + \bar{\lambda}^{2} \label{firstinequality}\\
& = -p\log z + \left(z + x + \bar{\lambda} \texttt{sign}(x)\right)^{2} - 2\bar{\lambda}z\,\texttt{sign}(x) \label{secondinequality}
\end{align}
Here, if $x \geq 0$, then $\bar{M} \geq -p\log z + z^2$ using the first inequality (\ref{firstinequality}), and if $x < 0$, then $\bar{M} \geq -p\log z +  2\bar{\lambda}z$ using the second inequality (\ref{secondinequality}). In either cases, the functions $-p\log z + z^2$ and $-p\log z + 2\bar{\lambda}z$ are unbounded as $z \rightarrow \infty$. Hence, the upper bound of $\bar{M}$ on these functions guarantee an upper bound $b_{M,\lambda,S}$ such that $|\omega_{ii}| \leq b_{M,\lambda,S}$. Therefore, $|\omega_{kk}| \leq b_{M,\lambda,S}$ for all $k = 1, 2, \ldots, p$.

(b) Lower bound: By positivity of the trace term and the $\ell_1$ term (for off-diagonals), we have
\begin{align}
M & \geq -\log\det{\Omega_D} = \sum_{i=1}^{i=p} -\log|\omega_{ii}|.
\end{align}
The negative log function $g(z) = -\log(z)$ is a convex function with a lower bound at $z^* = b_{M,\lambda,S}$ with $g(z^*) = -\log b_{M,\lambda,S}$.  Therefore, for any $k = 1, 2, \ldots, p$, we have 
\begin{align}
M &\geq \sum_{i=1}^{i=p} -\log|\omega_{ii}| \geq -(p-1)\log b_{M,\lambda,S} -\log|\omega_{kk}|.
\end{align}
Simplifying the above equation, we get the lower bound $a_{M,\lambda,S}$ on the diagonal entries $\omega_{kk}$. More specifically, 
\begin{align}
\log|\omega_{kk}| \geq - M - (p-1)\log b_{M,\lambda,S}.  \nonumber 
\end{align}
Therefore, $|\omega_{kk}| \geq a_{M,\lambda,S} = e^{- M - (p-1)\log b_{M,\lambda,S}} > 0$ serves as a lower bound for all $k = 1, 2, \ldots, p$.
\end{proof}

Given that the function values are non-increasing along the iterates of Algorithms \ref{concordmatrix:fbs}, \ref{concordmatrix:fista} and \ref{concordmatrix:pnopt}, the sequence of $\Omega^{(k)}$ satisfy $\Omega^{(k)} \in {\cal C}_{0}$ for $k = 1, 2, ....$. The lower bounds on the diagonal elements of $\Omega^{(k)}$ provides the Lipschitz continuity using
\begin{align}
\nabla^{2}h_1(\Omega^{(k)}) &\preceq \left(a_{M,\lambda,S}^{-2} + \|S\|_{2} \right) \left(I \otimes I\right).
\end{align}
Therefore, using the mean-value theorem, the gradient $\nabla{h_1}$ satisfies
\begin{align}
\|\nabla h_1(\Omega) - \nabla h_1(\Theta) \|_{F} & \leq L \|\Omega - \Theta\|_{F},
\end{align}
with the Lipschitz continuity constant $L = a_{M,\lambda,S}^{-2} + \|S\|_{2}$. The remaining argument for convergence follows from the theorems in \cite{beck2009fast}. 
\begin{theorem}
\label{1overk}
(\cite[Theorem 3.1]{beck2009fast}). Let $\{\Omega^{(k)}\}$ be the sequence generated by either Algorithm \ref{concordmatrix:fbs} with constant step size or with backtracking line-search. Then for any
$k \geq 1$,
\begin{alignat}{1}
Q_{\con}(\Omega^{(k)}) - Q_{\con}(\Omega^{*}) \leq \frac{\alpha L \|\Omega^{(0)} - \Omega^{*}\|_{F}^{2}}{2 k}
\end{alignat}
for the solution $\Omega^{*}$, where $\alpha = 1$ for the constant step size setting and $\alpha = c$ for the backtracking step size setting.
\end{theorem}

\begin{theorem}
\label{1overksq}
(\cite[Theorem 4.4]{beck2009fast}). For the sequences $\{\Omega^{(k)}\}, \{\Theta^{(k)}\}$ generated by Algorithm \ref{concordmatrix:fista}, for any $k \geq 1$, 
\begin{alignat}{1}
Q_{\con}(\Omega^{(k)}) - Q_{\con}(\Omega^{*}) \leq \frac{2 \alpha L \|\Omega^{(0)} - \Omega^{*}\|_{F}^{2}}{(k+1)^2}
\end{alignat}
for the solution $\Omega^{*}$, where $\alpha = 1$ for the constant step size setting and $\alpha = c$ for the backtracking step size setting. 
\end{theorem}

Hence, CONCORD-ISTA and CONCORD-FISTA converge at the rates of ${\cal O}(1/k)$ and ${\cal O}(1/k^{2})$ for the $k^{th}$ iteration.


\section{Implementation \& Numerical Experiments}
\label{sect: numerical}

In this section, we outline algorithm implementation details and
present results of our comprehensive numerical evaluation. Section
\ref{sect: synthetic} gives performance comparisons from using
synthetic multivariate Gaussian datasets. These datasets are generated
from a wide range of sample sizes ($n$) and dimensionality
($p$). Additionally, convergence of CONCORD-ISTA and CONCORD-FISTA
will be illustrated. Section \ref{sect: real} has timing results from
analyzing a real breast cancer dataset with outliers. Comparisons are
made to the coordinate-wise CONCORD implementation in {\tt gconcord}
package for R available at
\url{http://cran.r-project.org/web/packages/gconcord/}.

For implementing the proposed algorithms, we can take advantage of
existing linear algebra libraries. Most of the numerical computations
in Algorithms \ref{concordmatrix:fbs} and \ref{concordmatrix:fista}
are linear algebra operations, and, unlike the sequential
coordinate-wise CONCORD algorithm, CONCORD-ISTA and CONCORD-FISTA
implementations can solve increasingly larger problems as more and
more scalable and efficient linear algebra libraries are made
available. For this work, we opted to using Eigen library
\citep*{eigenweb} for its sparse linear algebra routines written in
C++. Algorithms \ref{concordmatrix:fbs} and \ref{concordmatrix:fista}
were also written in C++ then interfaced to R for testing. Table
\ref{tbl:naming convention} gives names for various CONCORD-ISTA and
CONCORD-FISTA versions using different initial step size choices.



\subsection{Synthetic Datasets}
\label{sect: synthetic}

Synthetic datasets were generated from true sparse positive random
$\Omega$ matrices of three sizes: $p=\{1000,\ 3000,\
5000\}$. Instances of random matrices used here consist of 4995, 14985
and 24975 non-zeros, corresponding to 1\%, 0.33\% and 0.20\% edge
densities, respectively. For each $p$, three random samples of sizes
$n=\{0.25p,\ 0.75p,\ 1.25p\}$ were used as inputs. The initial guess,
$\Omega^{(0)}$, and the convergence criteria was matched to those of
coordinate-wise CONCORD implementation. Highlights of the results are
summarized below, and the complete set of comparisons are given in
Supplementary materials Section \ref{appendix: timing comparison}.

For synthetic datasets, our experiments indicate that two variations
of the CONCORD-ISTA method show little performance
difference. However, {\tt ccista\_0} was marginally faster in our
tests. On the other hand, {\tt ccfista\_1} variation of CONCORD-FISTA
that uses $\tau_{(k+1,0)}=\tau_k$ as initial step size was
significantly faster than {\tt ccfista\_0}. Table
\ref{tbl:timing-short} gives actual running times for the two best
performing algorithms, {\tt ccista\_0} and {\tt ccfista\_1}, against
the coordinate-wise {\tt concord}. As $p$ and $n$ increase {\tt
  ccista\_0} performs very well. For smaller $n$ and $\lambda$,
coordinate-wise {\tt concord} performs well (more in Supplemental
section \ref{appendix: timing comparison}). This can be attributed to
$\min({\cal O}(np^2),{\cal O}(p^3))$ computational complexity of
coordinate-wise CONCORD \citep{Khare2013}, and the sparse linear
algebra routines used in CONCORD-ISTA and CONCORD-FISTA
implementations slowing down as the number of non-zero elements in
$\Omega$ increases. On the other hand, for large $n$ fraction
($n=1.25p$), the proposed methods {\tt ccista\_0} and {\tt ccfista\_1}
are significantly faster than coordinate-wise {\tt concord}. In
particular, when $p=5000$ and $n=6250$, the speed-up of {\tt
  ccista\_0} can be as much as 150 times over coordinate-wise {\tt
  concord}.





\begin{figure}[htb]
  \centering
\begin{knitrout}
\definecolor{shadecolor}{rgb}{0.969, 0.969, 0.969}\color{fgcolor}
\includegraphics[trim=0.3cm 0.7cm 0.5cm 0.5cm, clip, width=0.8\maxwidth]{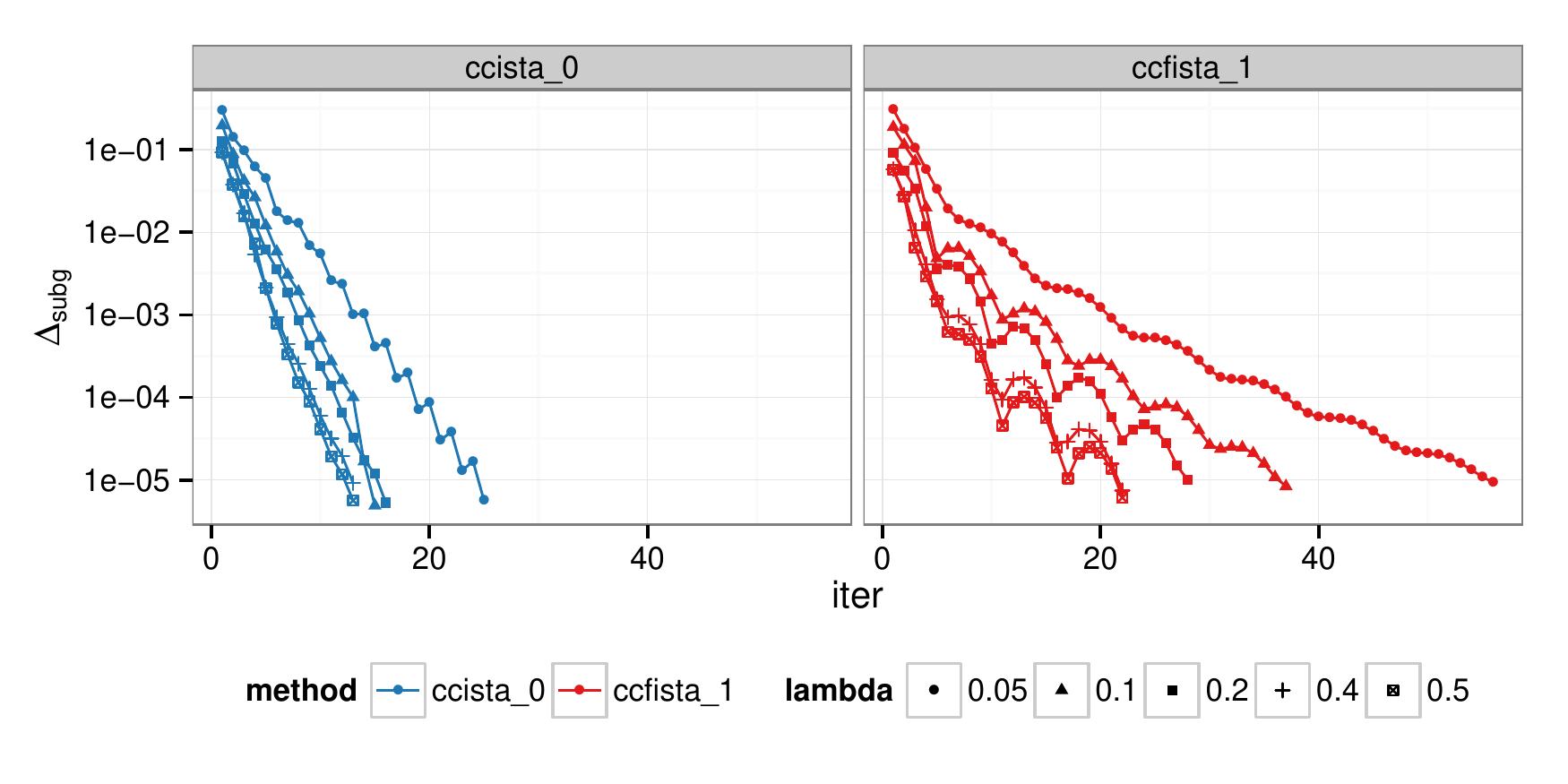} 
\end{knitrout}
\caption{Convergence of CONCORD-ISTA and CONCORD-FISTA. Convergence
  threshold is $\Delta_{\mathtt{subg}} < 10^{-5}$}
\label{fig:convergence}
\end{figure}

Convergence behavior of CONCORD-ISTA and CONCORD-FISTA methods is
shown in Figure \ref{fig:convergence}. The best performing algorithms
{\tt ccista\_0} and {\tt ccfista\_1} are shown. The vertical axis is
the subgradient $\Delta_{\texttt{subg}}$ (See Algorithms
\ref{concordmatrix:fbs}, \ref{concordmatrix:fista}). Plots show that
{\tt ccista\_0} seems to converge at a constant rate much faster than
{\tt ccfista\_1} that appears to slow down after a few initial
iterations. While the theoretical convergence results from section
\ref{sect: convergence} prove convergence rates of ${\cal O}(1/k)$ and
${\cal O}(1/k^2)$ for CONCORD-ISTA and CONCORD-FISTA, in practice,
{\tt ccista\_0} with constant step size performed the fastest for the
tests in this section.

\begin{table}[h]
  \centering
  \caption{Naming convention for step size variations} 
  \label{tbl:naming convention}
  \begin{tabular}{ccc}
    \hline
    \textbf{Variation} & \textbf{Method} & \textbf{Initial step size} \\ 
    \hline
    \ttbf concord & CONCORD & - \\ 
    \ttbf ccista\_0 & CONCORD-ISTA & Constant \\ 
    \ttbf ccista\_1 & CONCORD-ISTA & Barzilai-Borwein \\ 
    \ttbf ccfista\_0 & CONCORD-FISTA & Constant \\ 
    \ttbf ccfista\_1 & CONCORD-FISTA & $\tau_k$ \\ 
    \hline
  \end{tabular}
  \vspace{0.5cm}
  \caption{Timing comparison of {\tt concord} and proposed methods:
    {\tt ccista\_0} and {\tt ccfista\_1}.}
  \label{tbl:timing-short}
  \begin{tabular}{|c|c||c|r||rr||rr||rr||}\hline
    \multirow{2}{*}{$\bm p$} & 
    \multirow{2}{*}{$\bm n$} & 
    \multirow{2}{*}{$\bm \lambda$} & 
    \multirow{2}{*}{\bf NZ\%} & 
    \multicolumn{2}{c||}{\ttbf concord} & 
    \multicolumn{2}{c||}{\ttbf ccista\_0} & 
    \multicolumn{2}{c||}{\ttbf ccfista\_1} \\\cline{5-10}
    & & & & \bf iter & \bf seconds & \bf iter & \bf seconds & \bf iter & \bf seconds\\
    \hline
    \multirow{9}{*}{1000}
    & \multirow{3}{*}{250} 
      & 0.150 & 1.52 & 9 & 3.2 & 13 & \bf 1.8 & 20 & 3.3 \\ 
    & & 0.163 & 0.99 & 9 & 2.6 & 18 & \bf 2.0 & 26 & 3.3 \\ 
    & & 0.300 & 0.05 & 9 & 2.6 & 15 & \bf 1.2 & 23 & 2.7 \\ 
    \cline{2-10}
    & \multirow{3}{*}{750} 
      & 0.090 & 1.50 & 9 & 8.9 & 11 & \bf 1.4 & 17 & 2.5 \\ 
    & & 0.103 & 0.76 & 9 & 8.4 & 15 & \bf 1.6 & 24 & 3.3 \\ 
    & & 0.163 & 0.23 & 9 & 8.0 & 15 & \bf 1.6 & 24 & 2.8 \\ 
    \cline{2-10}
    & \multirow{3}{*}{1250} 
      & 0.071 & 1.41 & 9 & 41.3 & 10 & \bf 1.4 & 17 & 2.9 \\ 
    & & 0.077 & 0.97 & 9 & 40.5 & 15 & \bf 1.7 & 24 & 3.3 \\ 
    & & 0.163 & 0.23 & 9 & 43.8 & 13 & \bf 1.2 & 23 & 2.8 \\ 
    \hline
    \multirow{9}{*}{3000}
    & \multirow{3}{*}{750} 
      & 0.090 & 1.10 & 17 & 147.4 & 20 & \bf 32.4 & 25 & 53.2 \\ 
    & & 0.103 & 0.47 & 17 & 182.4 & 28 & \bf 36.0 & 35 & 60.1 \\ 
    & & 0.163 & 0.08 & 16 & 160.1 & 28 & \bf 28.3 & 26 & 39.9 \\ 
    \cline{2-10}
    & \multirow{3}{*}{2250} 
      & 0.053 & 1.07 & 16 & 388.3 & 17 & \bf 28.5 & 17 & 39.6 \\ 
    & & 0.059 & 0.56 & 16 & 435.0 & 28 & \bf 38.5 & 26 & 61.9 \\ 
    & & 0.090 & 0.16 & 16 & 379.4 & 16 & \bf 19.9 & 15 & 23.6 \\ 
    \cline{2-10}
    & \multirow{3}{*}{3750} 
      & 0.040 & 1.28 & 16 & 2854.2 & 17 & \bf 33.0 & 17 & 47.3 \\ 
    & & 0.053 & 0.28 & 16 & 2921.5 & 15 & \bf 23.5 & 16 & 31.4 \\ 
    & & 0.163 & 0.07 & 15 & 2780.5 & 25 & \bf 35.1 & 32 & 56.1 \\ 
    \hline
    \multirow{9}{*}{5000}
    & \multirow{3}{*}{1250}
      & 0.066 & 1.42 & 17 & 832.7 & 32 & \bf 193.9 & 37 & 379.2 \\ 
    & & 0.077 & 0.53 & 17 & 674.7 & 30 & \bf 121.4 & 35 & 265.8 \\ 
    & & 0.103 & 0.10 & 17 & 667.6 & 27 & \bf 81.2 & 33 & 163.0 \\ 
    \cline{2-10}
    & \multirow{3}{*}{3750} 
      & 0.039 & 1.36 & 17 & 2102.8 & 18 & \bf 113.0 & 17 & 176.3 \\ 
    & & 0.049 & 0.31 & 17 & 1826.6 & 16 & \bf 73.4 & 17 & 107.4 \\ 
    & & 0.077 & 0.10 & 17 & 2094.7 & 29 & \bf 95.8 & 33 & 178.1 \\ 
    \cline{2-10}
    & \multirow{3}{*}{6250}
      & 0.039 & 0.27 & 17 & 15629.3 & 17 & \bf 93.9 & 17 & 130.0 \\ 
    & & 0.077 & 0.10 & 17 & 15671.1 & 27 & \bf 101.0 & 25 & 123.9 \\ 
    & & 0.163 & 0.04 & 16 & 14787.8 & 26 & \bf 97.3 & 34 & 173.7 \\ 
    \hline
  \end{tabular}
  \vspace{0.5cm}
  \caption{Running time comparison on breast cancer dataset}
  \label{tbl:real data}
  \begin{tabular}{|c|r||rr||rr||rr||rr||rr||}\hline
    \multirow{2}{*}{$\bm \lambda$} & 
    \multirow{2}{*}{\bf NZ\%} & 
    \multicolumn{2}{c||}{\ttbf concord} & 
    \multicolumn{2}{c||}{\ttbf ccista\_0} & 
    \multicolumn{2}{c||}{\ttbf ccista\_1} & 
    \multicolumn{2}{c||}{\ttbf ccfista\_0} & 
    \multicolumn{2}{c||}{\ttbf ccfista\_1} \\\cline{3-12}
    & & \bf iter & \bf sec & \bf iter & \bf sec & \bf iter & \bf sec
    & \bf iter & \bf sec & \bf iter & \bf sec\\
  \hline
  0.450 & 0.110 & 80 & 724.5 & 132 & 686.7 & 123 & \bf 504.0 & 250 & 10870.3 & 201 & 672.6 \\ 
  0.451 & 0.109 & 80 & 664.2 & 129 & 669.2 & 112 & \bf 457.0 & 216 &  7867.2 & 199 & 662.9 \\ 
  0.454 & 0.106 & 80 & 690.3 & 130 & 686.2 &  81 & \bf 352.9 & 213 &  7704.2 & 198 & 677.8 \\ 
  0.462 & 0.101 & 79 & 671.6 & 125 & 640.4 & 109 & \bf 447.1 & 214 &  7978.4 & 196 & 646.3 \\ 
  0.478 & 0.088 & 77 & 663.3 & 117 & 558.6 &  87 & \bf 337.9 & 202 &  6913.1 & 197 & 609.0 \\ 
  0.515 & 0.063 & 63 & 600.6 & 104 & 466.0 &  75 & \bf 282.4 & 276 &  9706.9 & 184 & 542.0 \\ 
  0.602 & 0.027 & 46 & 383.5 &  80 & 308.0 &  66 & \bf 229.7 & 172 &  4685.2 & 152 & 409.1 \\ 
  0.800 & 0.002 & 24 & 193.6 &  45 & 133.8 &  32 & \bf  92.2 &  74 &  1077.2 &  70 & 169.8 \\ 
   \hline

  \end{tabular}
\end{table}

\subsection{Real Data}
\label{sect: real}

Real datasets arising from various physical and biological sciences
often are not multivariate Gaussian and can have outliers. Hence,
convergence characteristic may be different on such datasets. In this
section, the performance of proposed methods are assessed on a breast
cancer dataset \citep{Chang2005}. This dataset contains expression
levels of 24481 genes on 266 patients with breast cancer.  Following
the approach in Khare et al. \cite{Khare2013}, the number of genes are
reduced by utilizing clinical information that is provided together
with the microarray expression dataset. In particular, survival
analysis via univariate Cox regression with patient survival times is
used to select a subset of genes closely associated with breast
cancer. A choice of p-value $<0.03$ yields a reduced dataset with
$p=4433$ genes.

Often times, graphical model selection algorithms are applied in a
non-Gaussian and $n\ll p$ setting such as the case here. In this $n\ll
p$ setting, coordinate-wise CONCORD algorithm is especially fast due
to its computational complexity ${\cal O}(np^2)$. However, even in
this setting, the newly proposed methods {\tt ccista\_0}, {\tt
  ccista\_1}, and {\tt ccfista\_1} perform competitively to, or often
better than, {\tt concord} as illustrated in Table \ref{tbl:real
  data}. On this real dataset, {\tt ccista\_1} performed the fastest
whereas {\tt ccista\_0} was the fastest on synthetic datasets.

\section{Conclusion}
The Gaussian graphical model estimation or inverse covariance
estimation has seen tremendous advances in the past few years. In this
paper we propose using proximal gradient methods to solve the general
non-Gaussian sparse inverse covariance estimation problem. Rates of
convergence were established for the CONCORD-ISTA and CONCORD-FISTA
algorithms. Coordinate-wise minimization has been the standard
approach to this problem thus far, and we provide numerical results
comparing CONCORD-ISTA/FISTA and coordinate-wise minimization. We
demonstrate that CONCORD-ISTA outperforms coordinate-wise in general,
and in high dimensional settings CONCORD-ISTA can outperform
coordinate-wise optimization by orders of magnitude. The methodology
is also tested on real data sets. We undertake a comprehensive
treatment of the problem by also examining the dual formulation and
consider methods to maximize the dual objective. We note that efforts
similar to ours for the Gaussian case has appeared in not one, but
several NIPS and other publications. Our approach on the other hand
gives a complete and thorough treatment of the non-Gaussian partial
correlation graph estimation problem, all in this one self-contained
paper.


\clearpage
\bibliographystyle{unsrtnat}
\bibliography{library}

\begin{thebibliography}{12}
\providecommand{\natexlab}[1]{#1}
\providecommand{\url}[1]{\texttt{#1}}
\expandafter\ifx\csname urlstyle\endcsname\relax
  \providecommand{\doi}[1]{doi: #1}\else
  \providecommand{\doi}{doi: \begingroup \urlstyle{rm}\Url}\fi

\bibitem[Banerjee et~al.(2008)Banerjee, {El Ghaoui}, and
  D’Aspremont]{Banerjee2008}
Onureena Banerjee, Laurent {El Ghaoui}, and Alexandre D’Aspremont.
\newblock {Model Selection Through Sparse Maximum Likelihood Estimation for
  Multivariate Gaussian or Binary Data}.
\newblock \emph{JMLR}, 9:\penalty0 485--516, 2008.

\bibitem[Dalal and Rajaratnam(2014)]{Dalal2014}
Onkar~Anant Dalal and Bala Rajaratnam.
\newblock G-ama: Sparse gaussian graphical model estimation via alternating
  minimization.
\newblock \emph{arXiv preprint arXiv:1405.3034}, 2014.

\bibitem[Peng et~al.(2009)Peng, Wang, Zhou, and Zhu]{Peng2009}
Jie Peng, Pei Wang, Nengfeng Zhou, and Ji~Zhu.
\newblock {Partial Correlation Estimation by Joint Sparse Regression Models.}
\newblock \emph{Journal of the American Statistical Association}, 104\penalty0
  (486):\penalty0 735--746, June 2009.

\bibitem[Rocha et~al.(2008)Rocha, Zhao, and Yu]{Rocha2008}
Guilherme~V Rocha, Peng Zhao, and Bin Yu.
\newblock {A path following algorithm for Sparse Pseudo-Likelihood Inverse
  Covariance Estimation (SPLICE)}.
\newblock Technical Report 60628102, 2008.

\bibitem[Friedman et~al.(2010)Friedman, Hastie, and Tibshirani]{Friedman2010}
Jerome Friedman, Trevor Hastie, and Robert Tibshirani.
\newblock {Applications of the lasso and grouped lasso to the estimation of
  sparse graphical models}.
\newblock Technical report, 2010.

\bibitem[Khare et~al.(2014)Khare, Oh, and Rajaratnam]{Khare2013}
Kshitij Khare, Sang-Yun Oh, and Bala Rajaratnam.
\newblock A convex pseudo-likelihood framework for high dimensional partial
  correlation estimation with convergence guarantees.
\newblock \emph{to appear in Journal of the Royal Statistical Society: Series
  B}, 2014.

\bibitem[Beck and Teboulle(2009)]{beck2009fast}
Amir Beck and Marc Teboulle.
\newblock A fast iterative shrinkage-thresholding algorithm for linear inverse
  problems.
\newblock \emph{SIAM Journal on Imaging Sciences}, 2\penalty0 (1):\penalty0
  183--202, 2009.

\bibitem[Rockafellar(1976)]{rockafellar1976monotone}
R.T. Rockafellar.
\newblock Monotone operators and the proximal point algorithm.
\newblock \emph{SIAM Journal on Control and Optimization}, 14\penalty0
  (5):\penalty0 877--898, 1976.

\bibitem[Nesterov(1983)]{nesterov1983method}
Yurii Nesterov.
\newblock A method of solving a convex programming problem with convergence
  rate $\mathcal{O}(1/k^2)$.
\newblock In \emph{Soviet Mathematics Doklady}, volume~27, pages 372--376,
  1983.

\bibitem[Barzilai and Borwein(1988)]{barzilai1988two}
J.~Barzilai and J.M. Borwein.
\newblock Two-point step size gradient methods.
\newblock \emph{IMA Journal of Numerical Analysis}, 8\penalty0 (1):\penalty0
  141--148, 1988.

\bibitem[Guennebaud et~al.(2010)Guennebaud, Jacob, et~al.]{eigenweb}
Ga\"{e}l Guennebaud, Beno\^{i}t Jacob, et~al.
\newblock Eigen v3.
\newblock http://eigen.tuxfamily.org, 2010.

\bibitem[Chang et~al.(2005)Chang, Nuyten, Sneddon, Hastie, Tibshirani,
  S\o~rlie, Dai, He, van't Veer, Bartelink, van~de Rijn, Brown, and van~de
  Vijver]{Chang2005}
Howard~Y Chang, Dimitry S~A Nuyten, Julie~B Sneddon, Trevor Hastie, Robert
  Tibshirani, Therese S\o~rlie, Hongyue Dai, Yudong~D He, Laura~J van't Veer,
  Harry Bartelink, Matt van~de Rijn, Patrick~O Brown, and Marc~J van~de Vijver.
\newblock {Robustness, scalability, and integration of a wound-response gene
  expression signature in predicting breast cancer survival.}
\newblock \emph{Proceedings of the National Academy of Sciences of the United
  States of America}, 102\penalty0 (10):\penalty0 3738--43, March 2005.

\end{thebibliography}

\newpage
\appendix

\setcounter{page}{1}

\section*{Supplementary Materials}

\section{Timing comparison}
\label{appendix: timing comparison}

\subsection{Median Speed-up}

\begin{table}[H]
  \centering
  \caption{Median speed-up ratio over CONCORD method and (standard
    deviation).}
  \label{tbl:mean speedup}
  \begin{tabular}{|c|c||r|r||}
    \hline
    \multirow{2}{*}{$\bm p$}
    & \multirow{2}{*}{$\bm n$}
    & \multicolumn{2}{c||}{\bf Relative to {\ttbf concord}}\\\cline{3-4}
    & 
    & \multicolumn{1}{c|}{\ttbf ccista\_0}
    & \multicolumn{1}{c||}{\ttbf ccfista\_1}\\
  \hline
1000 & 250 &  0.6 ( 0.7) &  0.4 ( 0.3) \\ 
  1000 & 750 &  3.4 ( 1.8) &  1.9 ( 0.9) \\ 
  1000 & 1250 & 23.1 ( 5.7) & 12.0 ( 3.5) \\ 
   \hline
3000 & 750 &  2.7 ( 2.1) &  1.9 ( 1.6) \\ 
  3000 & 2250 & 12.8 ( 1.6) &  8.8 ( 2.2) \\ 
  3000 & 3750 & 81.9 ( 6.6) & 58.2 ( 8.7) \\ 
   \hline
5000 & 1250 &  5.6 ( 3.2) &  3.0 ( 1.8) \\ 
  5000 & 3750 & 21.1 ( 2.6) & 13.5 ( 2.6) \\ 
  5000 & 6250 & 145.8 ( 6.6) & 110.1 (16.4) \\ 
   \hline

\end{tabular}
\end{table}

\subsection{Comparison among CONCORD-ISTA and CONCORD-FISTA variations}

\begin{figure}[H]
  \centering
  \begin{subfigure}[b]{\textwidth}
\begin{knitrout}
\definecolor{shadecolor}{rgb}{0.969, 0.969, 0.969}\color{fgcolor}
\includegraphics[width=\maxwidth]{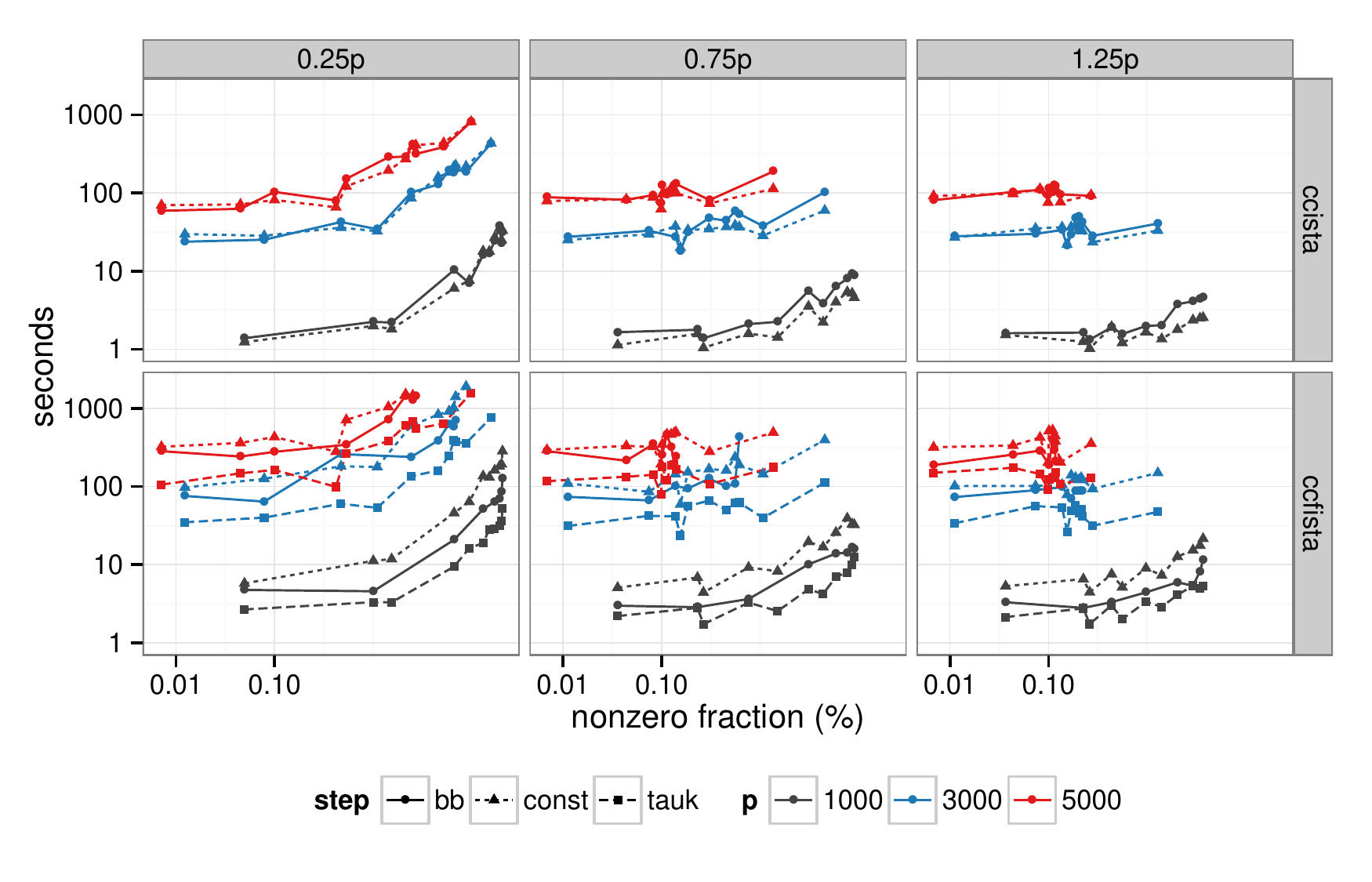} 

\end{knitrout}
\end{subfigure}
  \caption{Timings of CONCORD-ISTA (top) and CONCORD-FISTA (bottom) variations for
    sample sizes $n = \{0.25p, 0.75p, 1.25p\}$}
  \label{fig:compare-ccista-ccfista}
\end{figure}

\subsection{Comparison with CONCORD algorithm}

\begin{figure}[H]
  \centering
  \begin{subfigure}[b]{\textwidth}
\begin{knitrout}
\definecolor{shadecolor}{rgb}{0.969, 0.969, 0.969}\color{fgcolor}
\includegraphics[width=\maxwidth]{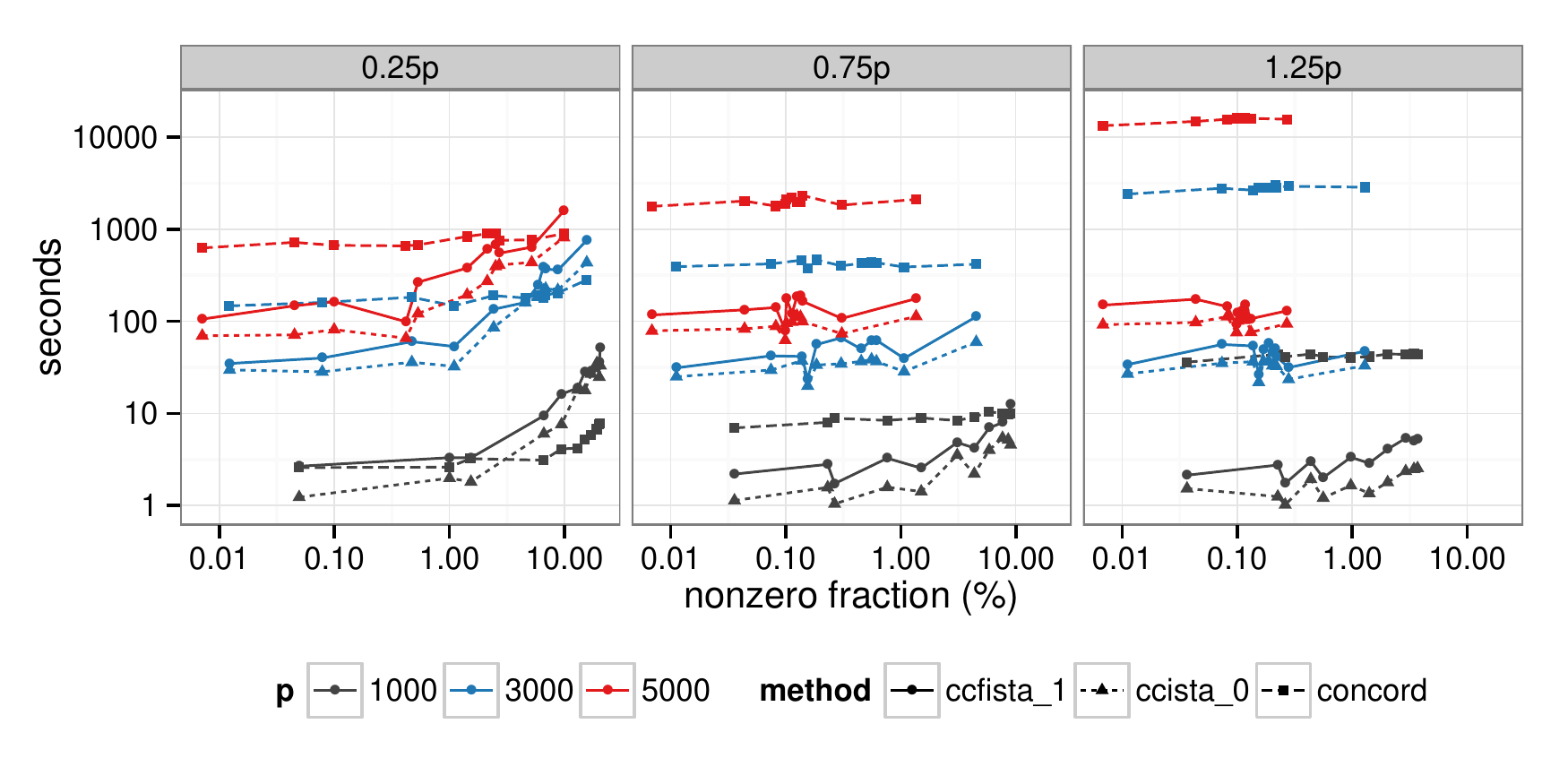} 

\end{knitrout}
\end{subfigure}
\caption{Timing of best CONCORD-ISTA and CONCORD-FISTA variations against CONCORD
  for sample sizes $n = \{0.25p, 0.75p, 1.25p\}$.}
  \label{fig:compare-ccistafista_all}
\end{figure}

\subsection{Running times}

\begin{table}[H]
  \centering
  \caption{$p=1000$, true non-zero fraction (nzf) of 1\%}
  \label{tab:p1000}
  \begin{tabular}{|c|c|c||c|r||rr||rr||rr||}\hline
    &\multirow{2}{*}{$p$} & 
    \multirow{2}{*}{$n$} & 
    \multirow{2}{*}{$\lambda$} & 
    \multirow{2}{*}{nzf (\%)} & 
    \multicolumn{2}{c||}{\bf concord} & 
    \multicolumn{2}{c||}{\bf ccista\_0} & 
    \multicolumn{2}{c||}{\bf ccfista\_1} \\\cline{6-11}
    & & & & & \bf iter & \bf seconds & \bf iter & \bf seconds & \bf iter & \bf seconds\\
  \hline
  7 & 1000 & 250 & 0.163 & 0.99 & 9 & 2.61 & 18 & 1.98 & 26 & 3.31 \\ 
  8 & 1000 & 250 & 0.300 & 0.05 & 9 & 2.58 & 15 & 1.23 & 23 & 2.67 \\ 
  9 & 1000 & 750 & 0.058 & 8.99 & 10 & 9.96 & 20 & 4.56 & 28 & 12.44 \\ 
  10 & 1000 & 750 & 0.059 & 8.56 & 10 & 9.86 & 20 & 5.19 & 28 & 9.86 \\ 
  11 & 1000 & 750 & 0.061 & 7.64 & 10 & 9.97 & 20 & 5.41 & 28 & 7.96 \\ 
  12 & 1000 & 750 & 0.066 & 5.86 & 10 & 10.45 & 20 & 4.01 & 27 & 6.96 \\ 
  13 & 1000 & 750 & 0.077 & 3.09 & 9 & 8.37 & 16 & 3.53 & 25 & 4.84 \\ 
  14 & 1000 & 750 & 0.103 & 0.76 & 9 & 8.40 & 15 & 1.58 & 24 & 3.26 \\ 
  15 & 1000 & 750 & 0.163 & 0.23 & 9 & 8.00 & 15 & 1.57 & 24 & 2.80 \\ 
  16 & 1000 & 750 & 0.300 & 0.04 & 8 & 6.96 & 13 & 1.13 & 18 & 2.20 \\ 
  17 & 1000 & 1250 & 0.058 & 3.69 & 9 & 44.21 & 15 & 2.54 & 24 & 5.29 \\ 
  18 & 1000 & 1250 & 0.059 & 3.43 & 9 & 44.25 & 16 & 2.49 & 24 & 5.00 \\ 
  19 & 1000 & 1250 & 0.061 & 2.91 & 9 & 43.84 & 16 & 2.36 & 24 & 5.38 \\ 
  20 & 1000 & 1250 & 0.066 & 2.03 & 9 & 44.15 & 14 & 1.79 & 24 & 4.09 \\ 
  21 & 1000 & 1250 & 0.077 & 0.97 & 9 & 40.50 & 15 & 1.65 & 24 & 3.34 \\ 
  22 & 1000 & 1250 & 0.103 & 0.44 & 9 & 44.16 & 15 & 1.93 & 24 & 3.02 \\ 
  23 & 1000 & 1250 & 0.163 & 0.23 & 9 & 43.84 & 13 & 1.25 & 23 & 2.75 \\ 
  24 & 1000 & 1250 & 0.300 & 0.04 & 8 & 35.99 & 13 & 1.53 & 17 & 2.13 \\ 
   \hline

  \end{tabular}
\end{table}

\begin{table}[H]
  \centering
  \caption{$p=3000$, true non-zero fraction (nzf) of 0.33\%}
  \label{tab:p3000}
  \begin{tabular}{|c|c|c||c|r||rr||rr||rr||}\hline
    &\multirow{2}{*}{$p$} & 
    \multirow{2}{*}{$n$} & 
    \multirow{2}{*}{$\lambda$} & 
    \multirow{2}{*}{nzf (\%)} & 
    \multicolumn{2}{c||}{\bf concord} & 
    \multicolumn{2}{c||}{\bf ccista\_0} & 
    \multicolumn{2}{c||}{\bf ccfista\_1} \\\cline{6-11}
    & & & & & \bf iter & \bf seconds & \bf iter & \bf seconds & \bf iter & \bf seconds\\
  \hline
  29 & 3000 & 750 & 0.077 & 2.42 & 18 & 190.00 & 36 & 85.81 & 31 & 135.13 \\ 
  30 & 3000 & 750 & 0.103 & 0.47 & 17 & 182.36 & 28 & 36.00 & 35 & 60.13 \\ 
  31 & 3000 & 750 & 0.163 & 0.08 & 16 & 160.13 & 28 & 28.29 & 26 & 39.94 \\ 
  32 & 3000 & 750 & 0.300 & 0.01 & 15 & 147.07 & 25 & 29.67 & 23 & 34.80 \\ 
  33 & 3000 & 2250 & 0.058 & 0.61 & 16 & 433.05 & 27 & 36.63 & 26 & 62.26 \\ 
  34 & 3000 & 2250 & 0.059 & 0.56 & 16 & 434.96 & 28 & 38.50 & 26 & 61.90 \\ 
  35 & 3000 & 2250 & 0.061 & 0.45 & 16 & 425.58 & 28 & 36.75 & 26 & 50.02 \\ 
  36 & 3000 & 2250 & 0.066 & 0.30 & 16 & 400.08 & 28 & 34.55 & 34 & 66.10 \\ 
  37 & 3000 & 2250 & 0.077 & 0.19 & 16 & 464.53 & 28 & 33.57 & 32 & 55.90 \\ 
  38 & 3000 & 2250 & 0.103 & 0.14 & 16 & 462.08 & 28 & 37.39 & 24 & 41.50 \\ 
  39 & 3000 & 2250 & 0.163 & 0.07 & 15 & 420.28 & 26 & 29.57 & 25 & 42.17 \\ 
  40 & 3000 & 2250 & 0.300 & 0.01 & 14 & 391.94 & 22 & 25.06 & 22 & 31.20 \\ 
  41 & 3000 & 3750 & 0.058 & 0.22 & 16 & 2837.71 & 27 & 32.61 & 24 & 41.36 \\ 
  42 & 3000 & 3750 & 0.059 & 0.21 & 16 & 2993.98 & 27 & 33.59 & 24 & 50.58 \\ 
  43 & 3000 & 3750 & 0.061 & 0.20 & 16 & 2826.17 & 27 & 33.06 & 24 & 45.75 \\ 
  44 & 3000 & 3750 & 0.066 & 0.19 & 16 & 2805.85 & 27 & 36.94 & 31 & 57.06 \\ 
  45 & 3000 & 3750 & 0.077 & 0.17 & 15 & 2792.55 & 26 & 36.61 & 31 & 48.96 \\ 
  46 & 3000 & 3750 & 0.103 & 0.14 & 15 & 2649.75 & 26 & 36.43 & 31 & 53.95 \\ 
  47 & 3000 & 3750 & 0.163 & 0.07 & 15 & 2780.53 & 25 & 35.12 & 32 & 56.06 \\ 
  48 & 3000 & 3750 & 0.300 & 0.01 & 13 & 2406.49 & 22 & 26.91 & 22 & 33.90 \\ 
   \hline
  \end{tabular}
  \vspace{0.5cm}
  \caption{$p=5000$, true non-zero fraction (nzf) of 0.20\%}
  \label{tab:p5000}
  \begin{tabular}{|c|c|c||c|r||rr||rr||rr||}\hline
    &\multirow{2}{*}{$p$} & 
    \multirow{2}{*}{$n$} & 
    \multirow{2}{*}{$\lambda$} & 
    \multirow{2}{*}{nzf (\%)} & 
    \multicolumn{2}{c||}{\bf concord} & 
    \multicolumn{2}{c||}{\bf ccista\_0} & 
    \multicolumn{2}{c||}{\bf ccfista\_1} \\\cline{6-11}
    & & & & & \bf iter & \bf seconds & \bf iter & \bf seconds & \bf iter & \bf seconds\\
  \hline
49 & 5000 & 1250 & 0.058 & 2.71 & 18 & 757.67 & 38 & 408.49 & 40 & 547.93 \\ 
  50 & 5000 & 1250 & 0.059 & 2.52 & 18 & 903.05 & 37 & 393.77 & 40 & 681.49 \\ 
  51 & 5000 & 1250 & 0.061 & 2.13 & 18 & 892.30 & 36 & 272.03 & 40 & 604.35 \\ 
  52 & 5000 & 1250 & 0.066 & 1.42 & 17 & 832.68 & 32 & 193.88 & 37 & 379.23 \\ 
  53 & 5000 & 1250 & 0.077 & 0.53 & 17 & 674.71 & 30 & 121.39 & 35 & 265.84 \\ 
  54 & 5000 & 1250 & 0.103 & 0.10 & 17 & 667.62 & 27 & 81.21 & 33 & 163.00 \\ 
  55 & 5000 & 1250 & 0.163 & 0.05 & 16 & 719.81 & 25 & 71.23 & 34 & 147.53 \\ 
  56 & 5000 & 1250 & 0.300 & 0.01 & 14 & 626.20 & 25 & 69.71 & 30 & 105.65 \\ 
  57 & 5000 & 3750 & 0.058 & 0.14 & 17 & 2324.54 & 29 & 99.50 & 35 & 165.12 \\ 
  58 & 5000 & 3750 & 0.059 & 0.13 & 17 & 1965.36 & 29 & 111.53 & 35 & 189.05 \\ 
  59 & 5000 & 3750 & 0.061 & 0.13 & 17 & 1967.39 & 29 & 114.72 & 35 & 186.34 \\ 
  60 & 5000 & 3750 & 0.066 & 0.11 & 17 & 2183.90 & 29 & 98.54 & 25 & 121.39 \\ 
  61 & 5000 & 3750 & 0.077 & 0.10 & 17 & 2094.73 & 29 & 95.84 & 33 & 178.13 \\ 
  62 & 5000 & 3750 & 0.103 & 0.08 & 16 & 1780.97 & 26 & 88.29 & 32 & 141.14 \\ 
  63 & 5000 & 3750 & 0.163 & 0.04 & 16 & 2021.49 & 25 & 82.88 & 33 & 133.36 \\ 
  64 & 5000 & 3750 & 0.300 & 0.01 & 14 & 1767.63 & 24 & 78.77 & 30 & 117.03 \\ 
  65 & 5000 & 6250 & 0.058 & 0.12 & 17 & 15698.02 & 27 & 113.65 & 25 & 150.95 \\ 
  66 & 5000 & 6250 & 0.059 & 0.12 & 17 & 16221.44 & 27 & 115.35 & 25 & 130.19 \\ 
  67 & 5000 & 6250 & 0.061 & 0.11 & 17 & 15698.53 & 27 & 103.06 & 25 & 132.57 \\ 
  68 & 5000 & 6250 & 0.066 & 0.11 & 17 & 16220.33 & 27 & 111.75 & 25 & 129.70 \\ 
  69 & 5000 & 6250 & 0.077 & 0.10 & 17 & 15671.14 & 27 & 101.03 & 25 & 123.92 \\ 
  70 & 5000 & 6250 & 0.103 & 0.08 & 17 & 15600.83 & 26 & 112.48 & 33 & 144.42 \\ 
  71 & 5000 & 6250 & 0.163 & 0.04 & 16 & 14787.78 & 26 & 97.33 & 34 & 173.66 \\ 
  72 & 5000 & 6250 & 0.300 & 0.01 & 14 & 13287.76 & 24 & 91.84 & 30 & 149.70 \\ 
   \hline

  \end{tabular}
\end{table}

\subsection{Other Methods}
\label{othermethods}
\subsubsection{Dual problem of CONCORD}
Formulating the dual using the matrix form is challenging since the KKT conditions involving the gradient term $S\Omega + \Omega S$ do not have a closed form solution as in the case of Gaussian problem in \cite{Dalal2014}. Therefore, we consider a vector form of the CONCORD problem by defining two new variables ${x_1} \in \mathbb{R}^{p}$ and ${x_2} \in \mathbb{R}^{{p(p-1)/2}}$ as
\begin{alignat}{1}
\label{variables}
{x_1} &= (\omega_{11}, \omega_{22}, \ldots , \omega_{pp})^{T} \nonumber \\
{x_2} &= (\omega_{12}, \omega_{13}, \ldots , \omega_{1p}, \omega_{23}, \ldots , \omega_{2p}, \ldots , \omega_{p-1p})^{T}.
\end{alignat}
We define two coefficient matrices ${A_1}, {A_2}$ as
\begin{alignat}{1}
\label{coefficients}
{A_1} = \displaystyle\left[\begin{array}{cccc} Y_{1} & & & \\ & Y_{2} & & \\ & & \ddots & \\ & & & Y_{p} \end{array}\right],
{A_2} = \displaystyle\left[
\begin{array}{ccccc} Y_{2} & Y_{3} & \cdots & Y_{p}\\ Y_{1} & & & & \\ & Y_{1} & & & \\ & & \ddots &  \\ & & & Y_{1} \end{array} 
\begin{array}{cccc} & & \\ Y_{3} & \cdots & Y_{p} \\ Y_{2} & & \\ & \ddots & \\ & & Y_{2} \end{array}\begin{array}{ccc} && \\ && \\ && \\ &\ddots&  \\ && \end{array} 
\begin{array}{cccc} & \\  &  \\  & \\  Y_{p-1} & Y_{p}\\ Y_{p-2} & \\ & Y_{p-2} \end{array} 
\begin{array}{c} \\  \\  \\  \\ Y_{p}\\ Y_{p-1} \end{array} 
\right],
\end{alignat}
where ${A_1}_{np \times p}$ and ${A_2}_{np \times p(p-1)/2}$ dimensional matrices. Using these definitions, the CONCORD problem (\ref{eq:matrix concord}) can be rewritten as
\begin{alignat}{1}
\label{cctransform}
\minimize_{{x_1}, {x_2}} \ & -n\log{x_1} + \displaystyle \frac{1}{2}\big\|{A_1}{x_1} + {A_2}{x_2}\big\|^{2} + \lambda \|{x_2}\|_{1}.
\end{alignat}
where, $\log({x_1}) = \sum_{i=1}^{i=p} \log({x_1}_{i})$. We will use $x = \left[\begin{array}{c}{x_1}\\ {x_2} \end{array}\right]$ for simplicity of notation where ever possible. 

The transformed CONCORD problem in (\ref{cctransform}) can be written in composite form using a new variable $z = {A_1}x_{1} + {A_2}x_{2}$ as 
\begin{alignat}{1}
\label{cccomposite}
\minimize_{x_{1},x_{2},z} \ & -n\log{x_{1}} + \displaystyle \frac{1}{2}\big\|z\big\|^{2} + \lambda \|x_{2}\|_{1} \nonumber \\
\st \ &{A_1}x_{1} + {A_2}x_{2} = z
\end{alignat}
The Lagrangian for this problem is given by
\begin{alignat}{1}
\label{cclagrange}
{\cal L}({x_1},{x_2},z, y)  &= -n\log{x_1} + \displaystyle \frac{1}{2}\big\|z\big\|^{2} + \lambda \|{x_2}\|_{1} + y^{T}\left({A_1}{x_1} + {A_2}{x_2} - z\right).
\end{alignat}
Maximizing with respect to the three primal variables yields following optimality conditions (the . notation is adapted from MATLAB to denote element-wise operations),
\begin{alignat}{1}
z - y  = 0 \nonumber\\
-n./{x_1} + {A_1}^{T}y  = 0 \nonumber \\
\lambda \mbox{sign}({x_2}) + {A_2}^{T}y \ni 0.
\end{alignat}
Substituting these the dual problem can be written as
\begin{alignat}{1}
\maximize_{y} \ & -n\log{n./{A_1}^{T}y} + \displaystyle \frac{1}{2}\big\|y\big\|^{2} + y^{T}\left({A_1}(n./{A_1}^{T}y) - y\right) \nonumber \\
\st \ &\|{A_2}^{T}y\|_{\infty} \leq \lambda, \nonumber
\end{alignat}
or equivalently
\begin{alignat}{1}
\label{concorddual}
\maximize_{y} \ & \displaystyle \frac{1}{2}\big\|y\big\|^{2} - n\log{({A_1}^{T}y)} + c \nonumber \\
\st \ &\|{A_2}^{T}y\|_{\infty} \leq \lambda,
\end{alignat}
where, $c = n\log{n} - n^{2}$ is a constant. This problem can also be written in composite form as 
\begin{alignat}{1}
\label{concorddual}
\maximize_{y} \ & \displaystyle \frac{1}{2}\big\|y\big\|^{2} - n\log{({A_1}^{T}y)} + \mathbbm{1}_{\|w\|_{\infty} \leq \lambda} \nonumber \\
\st \ & {A_2}^{T}y - w = 0.
\end{alignat}
The gradient and hessian of the smooth function $h(y) =  \frac{1}{2}\big\|y\big\|^{2} - n\log{({A_1}^{T}y)}$ is given by
\begin{alignat}{1}
\nabla{h}(y) &= y - {A_1}(n./{A_1}^{T}y), \nonumber\\
\nabla^{2}{h}(y) &= I + {A_1}\texttt{diag}\left(n./({A_1}^{T}y)^{2}\right){A_1}^{T}.
\end{alignat}
Here, the hessian is bounded away from the semi-definite boundary. Hence the function $h$ is strongly convex with parameter $1$. Moreover, on lines of Theorem \ref{up-low-bound}, we can show that if $y$ is restricted to a convex level set ${\cal C} = \left\{y | h(y) \leq M\right\}$ for some constant $M$, then the function $h$ has a Lipschitz continuous gradient. Note that
\begin{alignat}{1}
- n\log{({A_1}^{T}y)} & \leq {h}(y) \leq M \nonumber\\
e^{-\frac{M}{n}} &\leq {A_1}^{T}y.
\end{alignat}
Therefore, the hessian satisfies 
\begin{alignat}{1}
\nabla^{2}{h}(y) &= I + {A_1}\texttt{diag}\left(n./({A_1}^{T}y)^{2}\right){A_1}^{T} \preceq (1 + n\rho({A_1}^{T}{A_1}) e^{\frac{2M}{n}} )I.
\end{alignat}
To conclude, the dual problem provides an alternate method to prove the ${\cal O}(\frac{1}{k})$ and ${\cal O}(\frac{1}{k^2})$ rates of convergence for CONCORD problem.

\subsubsection{Proximal Newton's Algorithm for CONCORD}
Recall that the hessian of the smooth function $h_1$ as given in \ref{hessian} is
\begin{align}
\nabla^{2}{h_1}(\Omega) &= \sum_{i=1}^{i=p} {\omega_{ii}^{-2} \left[{e_i}{e_i}^{T} \otimes {e_i}{e_i}^{T}\right]} + \frac{1}{2} \left( S \otimes I + I \otimes S \right). \nonumber 
\end{align}
The subproblem solved for the direction of descent for the second order PNOPT algorithm is given by
\begin{alignat}{1}
\displaystyle {\Delta{\Omega}^{(k)}} = \argmin_{W} \la G^{(k)}, {W} \ra +  \frac{1}{2} \sum_{i=1}^{i=p} \omega_{ii}^{-2} \tr \left( W {e_i}{e_i}^{T}  W {e_i}{e_i}^{T} \right) + \tr \left( W S W \right) + \lambda \|{\Omega\X}^{(k)} + W\|_{1}.
\end{alignat}
Using these, the matrix version of the second order algorithm is given in Algorithm \ref{concordmatrix:pnopt}. Here, the subproblem for the descent step is as a huge Lasso problem. This can be solved by standard Lasso packages which uses coordinate descent methods.
\begin{algorithm}[H]
\caption{CONCORD - Proximal Newton Optimization Matrix form (CONCORD-PNOPT)}
\label{concordmatrix:pnopt}
\begin{algorithmic}
\STATE Initialize: $\Omega^{(0)} \in \mathbb{S}^{p}_{+}, \tau_{(0,0)} = 1, \Delta_{\texttt{opt}} = 2\epsilon_{\texttt{opt}}$ and $\Delta_{\texttt{term}} = 2\epsilon_{\texttt{term}}$

\WHILE{$\Delta_{\texttt{subg}} > \epsilon_{\texttt{subg}}$ \textbf{or} $\Delta_{\texttt{term}} > \epsilon_{\texttt{term}}$}  

  \STATE \textit{Compute $\nabla{h_1}$:}

  \STATE \qquad $G^{(k)} = {\Omega\D}^{-1} + \frac{1}{2}\left(S\,\,\Omega{(k)}^{T} + \Omega^{(k)}S\right)$

  \STATE \textit{Compute Newton step:}

  \STATE \qquad $\displaystyle {\Delta{\Omega}^{(k)}} = \argmin_{W} \la G^{(k)}, {W} \ra +  \frac{1}{2} \sum_{i=1}^{i=p} \omega_{ii}^{-2} \tr \left( W {e_i}{e_i}^{T}  W {e_i}{e_i}^{T} \right) + \tr \left( W S W \right) + \lambda \|{\Omega\X}^{(k)} + W\|_{1}$

  \STATE \textit{Compute sufficient descent} $\Delta^{(k)}$:

  \STATE \qquad $\Delta^{(k)} = \la G^{(k)}, {\Delta{\Omega}}^{(k)} \ra + \lambda\left(\| {\Omega\X}^{(k)} + \Delta{\Omega\X}^{(k)}\|_{1} - \| {\Omega\X}^{(k)} \|_{1} \right) $

  \STATE \textit{Compute $\tau_k$, such that $Q_{\con}(\Omega^{(k+1)}) \leq Q_{\con}(\Omega^{(k)}) + \alpha \tau_k \Delta^{(k)}$.}

  \STATE \textit{Update:} ${\Omega}^{(k+1)} = {\Omega}^{(k)} + \tau_{k} {\Delta{\Omega}}^{(k)}$

  \STATE \textit{Compute convergence criteria:} 

  \STATE \qquad $\Delta_{\texttt{subg}} = \displaystyle\frac{\| \nabla{h}(\Omega^{(k)}) + \partial{g}(\Omega^{(k)}) \|}{\|\Omega^{(k)}\|}$, \qquad $\Delta_{\texttt{term}} = \displaystyle\frac{\| f(\Omega^{(k+1)}) - f(\Omega^{(k)}) \|}{\|f(\Omega^{(k)})\|}$
\ENDWHILE
\end{algorithmic}
\end{algorithm}


\end{document}